\documentclass{article}
\usepackage[utf8]{inputenc}
\usepackage{amsmath, amssymb, amsthm, graphicx, hyperref, comment}
\usepackage{mathtools}
\usepackage{newunicodechar}
\usepackage{todonotes}
\usepackage{subcaption}
\newunicodechar{π}{\pi}

\newtheorem{lemma}{Lemma}

\newtheorem{theorem}{Theorem}
\newtheorem{definition}{Definition}

\newtheorem{proposition}{Proposition}

\title{Tighter Bounds for Algorithmic Complexity Estimation Using a Reusable Code-Based Block Decomposition Method}

\author{
{\small Eduardo Yuji Sakabe\textsuperscript{1,2},
Felipe S. Abrahão\textsuperscript{1,3,4,5},
Santiago Hernández-Orozco\textsuperscript{3}},\\
{\small Ricardo Gudwin\textsuperscript{2},
Hector Zenil\textsuperscript{1,3,6,7}}
}

\date{}

\begin{document}

\maketitle

\vspace{-2.0em}

\begin{center}
\scriptsize

\textsuperscript{1}Algorithmic Dynamics Lab, Karolinska Institute \& King's College London, London, UK \\

\textsuperscript{2}School of Electrical and Computer Engineering, University of Campinas (UNICAMP), Campinas, Brazil \\

\textsuperscript{3}Oxford Immune Algorithmics, Oxford University Innovation \& London Institute for Healthcare Engineering, UK \\

\textsuperscript{4}Centre for Logic, Epistemology and the History of Science, University of Campinas (UNICAMP), Campinas, Brazil \\

\textsuperscript{5}DEXL, National Laboratory for Scientific Computing (LNCC), Petrópolis, Brazil \\

\textsuperscript{6}Departments of Biomedical Computing and Digital Twins, School of Biomedical Engineering and Medical Sciences, King's Faculty of Life Sciences and Medicine, King's College London, UK \\

\textsuperscript{7}King's Institute for Artificial Intelligence, King's College London, UK
\end{center}

\vspace{1.0em}

\begin{abstract}
The Block Decomposition Method (BDM) was introduced as an alternative to popular lossless compression methods such as LZW for estimating algorithmic complexity from the principles of algorithmic probability and classical information theory. It extends the Coding Theorem Method (CTM) from small objects to larger ones by combining local estimates of algorithmic complexity with a global account of repetition based on Shannon entropy. Here, we introduce a version of BDM in which dependencies between blocks are utilized to reduce the length of the description based on reusable program code in the decomposition of an object, and on conditional descriptions capable of accounting for shared structure between observations. We formalize this allocation of descriptive resources as \emph{algorithmic attention}. Repeated or related components need not be described independently, and the resulting reduction in description length is governed by the amount of shared algorithmic information. We formulate this extension as a reuse optimization problem, show that exact optimization is NP-hard, derive conditions under which it improves upon independent descriptions, relate the achievable gains to algorithmic mutual information, prove the relationship with the previous BDM version, and provide a roadmap for its implementation using CTM-derived complexity and conditional complexity estimates.\\

\end{abstract}

\section{Introduction}

Algorithmic information theory (AIT) provides a formal account of structure in terms of description length: an object is simple when it admits a short description, and random when no such description exists. This perspective is especially relevant when regularities arise from computable mechanisms rather than from statistical repetition alone. However, Kolmogorov complexity is uncomputable. Practical estimators based on algorithmic probability, such as the Coding Theorem Method (CTM)~\cite{delahaye_numerical_2012}, can capture algorithmic regularities not reflected by Shannon entropy or lossless statistical compressors such as LZW, but they are restricted to small objects. The Block Decomposition Method (BDM) was introduced to scale CTM-based estimation by decomposing larger objects into small blocks, assigning local algorithmic-complexity estimates, and accounting for repeated blocks through logarithmic multiplicity terms~\cite{zenil_decomposition_2018,zenil_review_2020}.

BDM has since been applied across artificial intelligence, machine learning, causal analysis, network science, and human behaviour~\cite{hernandez-espinosa_superarc_2026,sakabe_binarized_2026,bakhtiarifard_characterizing_2026,hernandez-orozco_algorithmic_2021,zenil_causal_2019,zenil_minimal_2025,gauvrit_human_2017}. Yet its original formulation retains a structural limitation. Its local CTM-based estimates capture algorithmic structure, while aggregation across blocks depends only on multiplicity counts. BDM therefore accounts for exact repetition but assigns independent descriptions to distinct blocks, even when they share algorithmic information through short transformations, common generators, or reusable subprograms.

We introduce BDM\,2.0 to incorporate such dependencies into decomposition-based estimation. Rather than assigning independent descriptions to all distinct blocks, BDM\,2.0 allows blocks, or their generating programs, to be described conditionally from one another. We formalize this allocation of descriptive resources as \emph{algorithmic attention}, defined as the selection of explanatory components from which the largest reduction in total description length can be obtained. BDM is recovered as the independent-description limit, up to representation overhead. When reusable structure is present, BDM\,2.0 replaces independent descriptions with conditional ones and improves whenever the estimated reuse gain exceeds the relevant representation and validation overheads. The source of this gain is algorithmic mutual information between blocks, programs, or both.

BDM\,2.0 extends the class of structure that BDM can make empirically accessible. In many systems, regularity is not exhausted by the complexity or exact repetition of individual components. It also lies in short mechanisms that generate or transform one component from another. BDM\,2.0 captures these dependencies through conditional descriptions whenever they are shorter than independent descriptions. Shared generators can therefore reduce the estimate when BDM\,1.0 overpays by treating distinct blocks as independent. This is the regime in which AIT is most informative: when explanation is tied to causal structure and universal compression rather than to statistical fit alone.

\section{Background}

\subsection{Kolmogorov Complexity and Algorithmic Probability}

We briefly recall the basic notions of algorithmic information theory \cite{kolmogorov_three_1968, chaitin_length_1969, solomonoff_formal1_1964, solomonoff_formal2_1964, li_introduction_2019}.

Let \(U\) be a fixed universal prefix-free Turing machine. The \emph{Kolmogorov complexity}, also known as algorithmic or Kolmogorov--Chaitin complexity, of a string \(x\) is
\[
K(x) = \min\{|p| : U(p)=x\}.
\]
An object \(x\) is \emph{algorithmically random}, or incompressible up to a constant, if
\[
K(x) \ge |x| - c
\]
for some fixed constant \(c\). In this case, \(x\) has no substantially shorter description than itself. The constant \(c\) reflects the choice of the reference universal machine and, by the invariance theorem, does not depend on \(|x|\).

The \emph{conditional Kolmogorov complexity} of \(x\) given \(y\) is
\[
K(x \mid y) = \min\{\,|p| : U(p,y) = x \,\},
\]
representing the minimal information required to generate \(x\) when \(y\) is provided as auxiliary input.

The \emph{algorithmic probability} of \(x\), also called the universal a priori semimeasure, is
\[
m(x) = \sum_{p : U(p) = x} 2^{-|p|}.
\]
It is the probability that a universal prefix-free machine outputs \(x\) when fed with a random program whose bits are chosen independently with probability \(1/2\).

\emph{Levin's coding theorem} \cite{levin_laws_1974} establishes a fundamental connection between algorithmic probability and complexity:
\[
K(x) = -\log m(x) + O(1),
\]
where logarithms are taken in base two. Thus, objects with high algorithmic probability have low Kolmogorov complexity, and conversely.

The \emph{algorithmic mutual information} between \(x\) and \(y\) is
\[
I(x:y)
=
K(x)+K(y)-K(x,y)
=
K(x)-K(x\mid y^*)+O(1),
\]
where \(y^*\) is a shortest description of \(y\). It quantifies their shared algorithmic structure, or equivalently the extent to which knowing one object reduces the description length of the other.

\subsection{The Coding Theorem Method (CTM)}

The \emph{Coding Theorem Method} (CTM) \cite{delahaye_numerical_2012, soler-toscano_calculating_2014} approximates \(K(x)\) by empirically estimating \(m(x)\) from the output distribution of small Turing machines. For \(n\)-state, two-symbol Turing machines, the cardinality of the corresponding machine space is
\[
N(n)=(4n+2)^{2n},
\]
as detailed in Appendix~\ref{app:ctm_counts}. Existing CTM tables have been obtained by exhaustive or reduced enumeration of small machine spaces, most notably the 5-state, 2-symbol case. Larger spaces, such as \(n=6\), require partial enumeration or sampling.

CTM is grounded in algorithmic probability and can be instantiated across different computational formalisms. We use Turing machines for illustration and adopt a binary alphabet, since any finite alphabet can be encoded in binary. CTM therefore provides an algorithmic prior over small strings and discrete objects.

Formally, CTM approximates Kolmogorov complexity by leveraging the relation between algorithmic probability and complexity established by Levin's coding theorem \cite{levin_laws_1974, levin_concrete_1977}. Let \(D_n(s)\) denote the empirical frequency of a string \(s\) among the outputs of all halting \(n\)-state, 2-symbol Turing machines in the chosen enumeration. Then, by the coding theorem, the complexity of \(s\) can be approximated as
\[
\mathrm{CTM}_n(s) \approx -\log_2 D_n(s).
\]
In this way, CTM estimates \(D_n\) by enumeration of small machines and derives an approximation to \(K(s)\) from the observed output distribution.

\subsection{Conditional CTM}

Standard CTM assumes that Turing machines begin from a blank input tape. In this work, we extend this framework to estimate the \emph{conditional} algorithmic complexity \(K(y \mid x)\), defined as the length of the shortest program that, when given input \(x\), produces output \(y\):
\[
K(y \mid x) = \min_{p} \{ |p| : U(p, x) = y \}.
\]
To approximate this quantity using CTM, we modify the enumeration so that machines are executed on a fixed input tape configuration \(x\), and we count the programs that halt with output \(y\). This induces the conditional algorithmic probability
\[
m(y \mid x) = \sum_{p : U(p, x) = y} 2^{-|p|}.
\]

As in standard CTM, we estimate this probability empirically from the output distribution of the enumerated machines. Let \(D_n(y \mid x)\) denote the empirical frequency with which \(y\) is produced by halting \(n\)-state machines initialized with input \(x\). The conditional CTM estimate is then
\[
\mathrm{CTM}_n(y \mid x) \approx -\log_2 D_n(y \mid x).
\]

Precomputing conditional CTM values for all input--output pairs would require estimating an empirical output distribution for every possible input tape configuration. This is prohibitively expensive beyond very small domains, since the relevant machine space must be explored relative to each input. In practice, conditional CTM can be estimated by partial enumeration of the Turing-machine space for a fixed input \(x\). A complementary approach is to exploit execution traces from enumerated machines, using observed transitions between intermediate configurations as empirical evidence for conditional descriptions. This provides a local, context-dependent estimate of the algorithmic information required to transform \(x\) into \(y\).

\subsection{The Block Decomposition Method (BDM)}

For larger objects \(X\), CTM computation is infeasible.
The \emph{Block Decomposition Method} (BDM) decomposes \(X\) into small blocks and groups them by distinct block types \(x_i\), each lying within the domain of a computable complexity estimator:
\[
\mathrm{BDM1}(X) = \sum_{i} \left( \hat K(x_i) + \log m_i \right),
\]
where $m_i$ is the multiplicity of block $x_i$.

Notice that $\hat K$ can be instantiated by different approximations of Kolmogorov complexity and not necessarily by CTM. Thus, BDM is agnostic to the actual method and is rather a means to combine different approaches. In standard BDM, this estimator is instantiated by CTM:
\[
\hat K(x_i) = \mathrm{CTM}(x_i).
\]
Yet, being possibly the only alternative to Shannon-entropy-based methods, including popular lossless statistical compression methods such as LZW~\cite{ziv_universal_1977, welch_technique_1984} and cognates (ZIP, PNG, etc.), CTM is grounded in algorithmic probability and can capture algorithmic regularities beyond statistical repetition~\cite{zenil_review_2020}.

The \(\log m_i\) term prevents repeated occurrences of the same block from being charged as independent descriptions by accounting instead for their multiplicity. For a fixed block size and a fixed table of \(\hat K\)-values, the distinct-block complexity contribution remains bounded, and for increasingly large objects the discrepancy between BDM and the corresponding Shannon block entropy becomes negligible relative to the total number of blocks~\cite{zenil_decomposition_2018}.

This ability to estimate the algorithmic complexity of large empirical objects has supported applications across several domains. In artificial intelligence and machine learning, BDM has been used in SuperARC as a human-agnostic benchmark of model abstraction and prediction under increasing algorithmic complexity~\cite{hernandez-espinosa_superarc_2026}; in binarized neural networks to show that algorithmic complexity tracks training dynamics more closely than entropy~\cite{sakabe_binarized_2026}; and in full-precision deep neural networks to characterize generalization, overfitting, and grokking dynamics while informing pruning strategies~\cite{bakhtiarifard_characterizing_2026}. BDM and related algorithmic-probability methods have also supported learning over non-differentiable spaces~\cite{hernandez-orozco_algorithmic_2021}, causal deconvolution into likely generative mechanisms~\cite{zenil_causal_2019}, minimal-information-loss dimensionality reduction and network sparsification~\cite{zenil_minimal_2025}, and studies of behavioural complexity across human development~\cite{gauvrit_human_2017}.

Hereafter, we refer to this original formulation as BDM\,1.0.

\section[BDM 2.0: Algorithmic Attention and Reuse]{BDM\,2.0: Algorithmic Attention and Reuse}

BDM\,2.0 extends the original Block Decomposition Method (BDM) by recognizing that distinct blocks need not correspond to independent descriptions. When two blocks, or their generating programs, share algorithmic information, part of their description can be reused rather than paid for repeatedly. BDM\,2.0 therefore seeks the shortest explanatory representation available under the chosen estimators by exploiting algorithmic mutual information in either observation space (shared outputs) or program space (shared generators). This extension measures not merely similarities between observations, but the shared \emph{causal generators} responsible for their structural regularities.

BDM\,2.0 assumes a decomposition \(B(X)=\{x_1,\ldots,x_n\}\) of an object \(X\) into distinct blocks with multiplicities \(m_1,\ldots,m_n\), together with a finite, nonempty set of candidate programs \(\mathcal{P}(x_i)\) capable of generating each block \(x_i\). Each candidate set is assumed to contain at least one shortest or near-shortest descriptor \(p_i^*\) of \(x_i\), so that \(p_i^*\in\mathcal{P}(x_i)\) and \(U(p_i^*)=x_i\). The method further assumes access to computable estimators \(\hat K(\cdot)\) and \(\hat K(\cdot\mid\cdot)\) of Kolmogorov and conditional Kolmogorov complexity, respectively, such as those obtained through the Coding Theorem Method (CTM).

Let
\[
\mathcal{E}_X
=
\left\{
P=\{p_1,\ldots,p_n\}
:
p_i\in\mathcal{P}(x_i)
\right\}
\]
denote the family of explanatory program sets containing one candidate program for each distinct block of \(X\). Since \(U\) is deterministic and the blocks are distinct, the selected programs are also distinct, so \(P\) is well defined as a set. Enlarging any \(\mathcal{P}(x_i)\) enlarges \(\mathcal{E}_X\) and expands the admissible choices available to the reuse optimization.

The method selects an explanatory set \(P\in\mathcal{E}_X\) together with a nonempty subset \(S\subseteq P\) of explanatory programs from which the remaining programs can be most compactly described. Each program in \(S\) is paid for once through its estimated complexity, while every remaining program \(p_j\in P\setminus S\) is described relative to a program \(p_i\in S\). The cost of explaining \(p_j\) is then given by the smallest available reuse cost, either in program space through \(\hat K(p_j\mid p_i)\) or in observation space through \(\hat K(x_j\mid x_i)\).

The resulting \emph{explanation cost} is therefore defined as
\begin{equation}
\label{eq:explanation-cost}
\begin{aligned}
\mathcal{L}_{\mathrm{expl}}(X)
=
\min_{P\in\mathcal{E}_X}
\min_{\varnothing\neq S\subseteq P}
\Bigg(
&\sum_{p_i\in S} \hat K(p_i) \\
&+
\sum_{p_j\in P\setminus S}
\min_{p_i\in S}
\min\left\{
\hat K(p_j\mid p_i),
\hat K(x_j\mid x_i)
\right\}
\Bigg).
\end{aligned}
\end{equation}
The first term describes the representative programs in \(S\), whose execution produces the corresponding blocks \(x_i\). The second term uses these representatives and their outputs to reconstruct each remaining target block \(x_j\). The program-level conditional \(\hat K(p_j\mid p_i)\) recovers \(p_j\) from \(p_i\), after which executing \(p_j\) produces \(x_j\). Alternatively, since executing \(p_i\) already provides \(x_i\), the observation-level conditional \(\hat K(x_j\mid x_i)\) recovers \(x_j\) directly, without reconstructing \(p_j\). Equation~\eqref{eq:explanation-cost} therefore provides a structured estimate of the joint complexity of the blocks by combining independently described representatives with conditional descriptions of the remaining components, up to representation, decoding, and indexing overheads.

Following the general definition of attention as a system that allocates structural or temporal resources to perform a task~\cite{santana_attention_2022,santana_neural_2021}, we introduce this principle of selective allocation within the framework of Algorithmic Information Dynamics (AID)~\cite{zenil_algorithmic_2023} as \emph{algorithmic attention}: the allocation of descriptive resources to the subset \(S\) of explanations that yields the largest reduction in total description length. Under this view, the relevant components of an object are not necessarily those that are locally most complex, but those whose description enables other observations or programs to be recovered at lower conditional cost. Descriptive resources are therefore not assigned uniformly across all observations, but preferentially to the explanatory components from which other observations or programs can be derived.

The analogy is functional. As in neural attention~\cite{santana_neural_2021}, including self-attention in Transformer models~\cite{vaswani_attention_2017}, some components are prioritized because they provide useful context for processing others. In algorithmic attention, this context is provided by shared generators or conditional descriptions, allowing related observations or programs to be described at lower description cost.

Figure~\ref{fig:bdm2-concept} schematically summarizes this shift from independent block descriptions in BDM\,1.0 to algorithmic attention and reuse in BDM\,2.0.

\begin{figure}[!t]
\centering
\includegraphics[width=\linewidth]{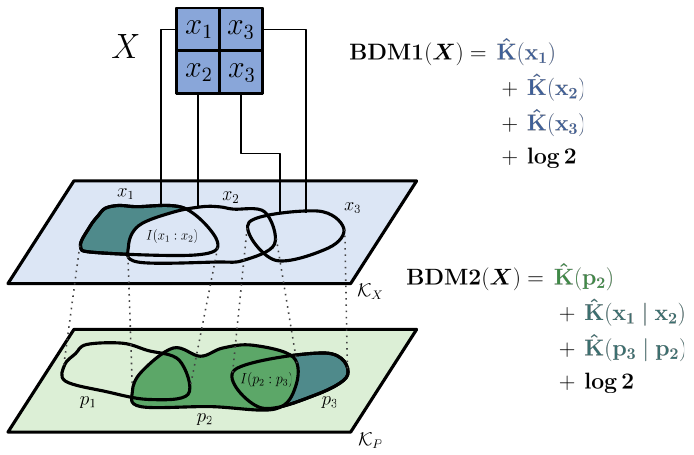}
\caption{\textbf{BDM\,1.0 independent descriptions versus BDM\,2.0 algorithmic attention and reuse.}
BDM\,1.0 operates at the level of distinct blocks, treating their descriptions as independent and summing the estimated complexities of each unique block. BDM\,2.0 extends this view by allowing dependencies to be represented in both program-complexity space (\(\mathcal{K}_P\)) and object-complexity space (\(\mathcal{K}_X\)). The quantities \(I(p_2:p_3)\) and \(I(x_1:x_2)\) illustrate shared algorithmic information between programs and observations, respectively. BDM\,2.0 exploits such mutual information for reuse, through algorithmic attention, by selecting explanatory programs and using conditional descriptions whenever they reduce the total description length, thereby seeking a minimum-cost explanation of the object. In both methods, repeated blocks are accounted for through the logarithmic multiplicity term \(\log m_i\). The figure is schematic and intended only to convey intuition.}
\label{fig:bdm2-concept}
\end{figure}

Although \(\mathcal{P}(x_i)\) contains shortest or near-shortest generators as an admissible baseline, it is not restricted to them. A non-minimal program \(u\in\mathcal{P}(x_i)\), with \(U(u)=x_i\), can still be selected into \(S\) if it provides short conditional descriptions of many other programs or observations. More precisely, for a target set \(T\subseteq\{1,\ldots,n\}\setminus\{i\}\), such a program becomes advantageous over independent descriptions whenever
\[
\hat K(u)
+
\sum_{j\in T}
\min\left\{
\hat K(p_j\mid u),
\hat K(x_j\mid x_i)
\right\}
<
\hat K(p_i)
+
\sum_{j\in T}\hat K(p_j),
\]
up to fixed representation and decoding overheads. Thus, the optimization does not privilege minimality alone, but the trade-off between the cost of describing a representative program and the savings obtained by reusing it.

This reuse trade-off motivates the conjecture that highly reprogrammable programs yield larger reductions in description length through reuse in BDM\,2.0. Reprogrammability characterizes a program's capacity to emulate a range of other programs through coarse-grained transformations~\cite{zenil_asymptotic_2017}. Programs with high reprogrammability, such as Busy Beaver machines studied as candidates for high emulation capacity, may have larger unconditional description length while carrying substantial algorithmic mutual information with many other programs. Such programs are natural candidates for \(S\). They need not be the shortest explanation of any single block, but may serve as algorithmic hubs from which multiple explanations can be recovered at low conditional cost.

One might ask why program-level reuse is needed rather than relying only on conditional descriptions between observations. At the limit of true Kolmogorov complexity, output-level conditionals can in principle capture any computable dependency between \(x_i\) and \(x_j\). BDM\,2.0, however, relies on finite computable estimators \(\hat K\), whose conditional estimates depend on the chosen representation, model class, and search budget. Thus, a short transformation from \(p_i\) to \(p_j\) need not be detected as a short estimated description from \(x_i\) to \(x_j\). As illustrated in Figure~\ref{fig:observation-program-conditionals}, simple transformations may be directly detectable between observations, whereas closely related programs may generate outputs whose dependency is difficult to recover in observation space. We therefore retain conditional descriptions in both observation and program space.

\begin{figure}[!t]
\centering

\begin{subfigure}[t]{\textwidth}
    \centering
    \includegraphics[width=0.88\linewidth]{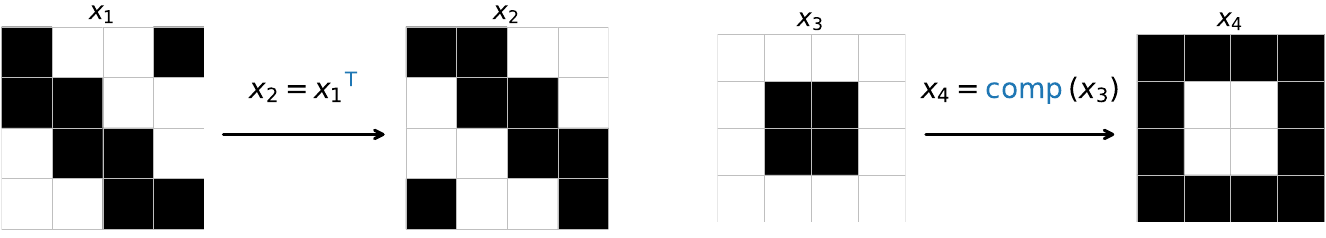}
    \caption{Fixed transformations between observations.}
    \label{fig:observation-space-conditionals}
\end{subfigure}

\vspace{0.4em}

\begin{subfigure}[t]{\textwidth}
    \centering
    \includegraphics[width=\linewidth]{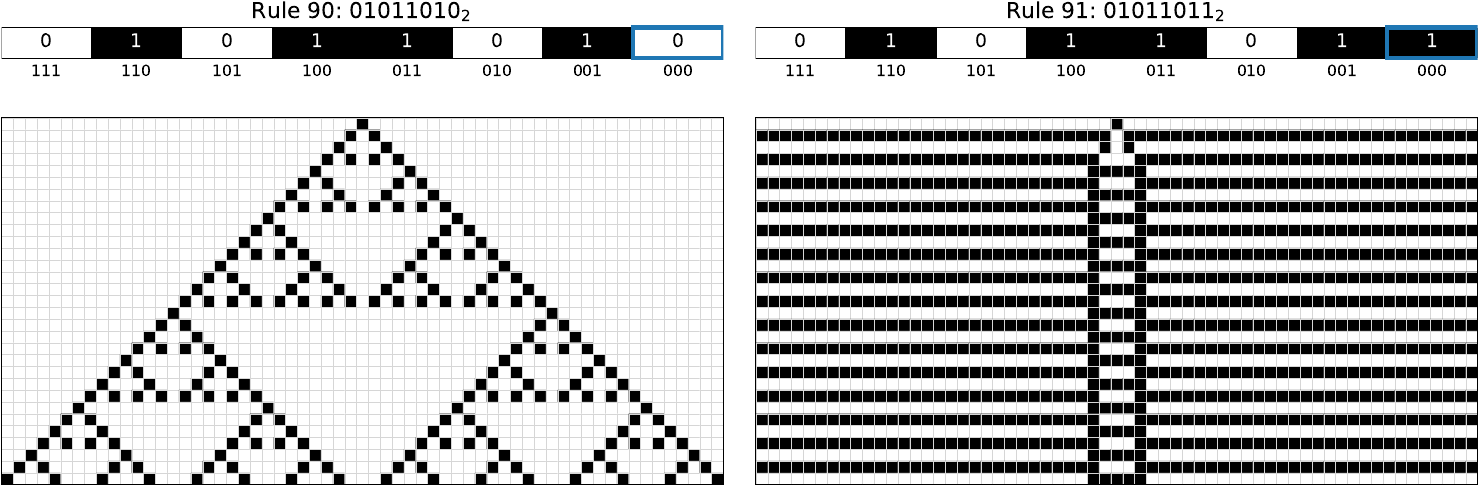}
    \caption{A one-bit transformation between programs.}
    \label{fig:program-space-conditionals}
\end{subfigure}
\caption{\textbf{Conditional descriptions can be short in observation or program space.} \textbf{(a)} Distinct \(4\times4\) blocks satisfy \(x_2=x_1^{\mathsf T}\) and \(x_4=\operatorname{comp}(x_3)\). Because transposition and bitwise complementation are fixed computable transformations, the observation-level costs \(\hat K(x_2\mid x_1)\) and \(\hat K(x_4\mid x_3)\) can be small when these operations are represented by the estimator. \textbf{(b)} ECA Rules \(90=01011010_2\) and \(91=01011011_2\) differ only in the output assigned to neighbourhood \(000\), highlighted in blue. With the simulator, initial condition, and boundary conditions fixed, the corresponding programs \(p_{90}\) and \(p_{91}\) differ by one rule-table bit, so \(\hat K(p_{91}\mid p_{90})\) can be small even when the difference between their resulting spacetime diagrams is large and no comparably simple observation-level transformation relating them is apparent. For each distinct target block, BDM\,2.0 compares the available observation-level and program-level conditional descriptions and selects the lower estimated cost.}
\label{fig:observation-program-conditionals}
\end{figure}

As in BDM\,1.0, repeated occurrences of the same block contribute only logarithmically. The multiplicity correction is given by
\[
\mathcal{L}_{\mathrm{mult}}(X)
=
\sum_{i=1}^{n}\log m_i.
\]

The final BDM\,2.0 estimate is therefore
\[
\mathrm{BDM2}(X)
=
\mathcal{L}_{\mathrm{expl}}(X)
+
\mathcal{L}_{\mathrm{mult}}(X).
\]

Equivalently,
\begin{equation}
\label{eq:bdm2}
\begin{aligned}
\mathrm{BDM2}(X)
=
&
\min_{P\in\mathcal{E}_X}
\min_{\varnothing\neq S\subseteq P}
\Bigg[
\sum_{p_i\in S}\hat K(p_i)
+
\sum_{p_j\in P\setminus S}
\min_{p_i\in S}
\min
\left\{
\begin{aligned}
&\hat K(p_j\mid p_i),\\
&\hat K(x_j\mid x_i)
\end{aligned}
\right\}
\Bigg]
\\
&
+
\sum_{i=1}^{n}\log m_i.
\end{aligned}
\end{equation}

Under this formulation, BDM\,2.0 replaces a decomposition over observations with a decomposition over explanations.

\subsection[BDM 1.0 as an Independent Limit and Reuse Gain]{BDM\,1.0 as an Independent Limit and Reuse Gain}
\label{sec:bdm1-independent-limit}

In this subsection, we establish that BDM\,1.0 is the independent-description limit of BDM\,2.0, obtained when each distinct block is described separately and no reuse is exploited. Choosing this no-reuse construction is always an admissible BDM\,2.0 solution, so BDM\,2.0 cannot exceed the independent-description bound except for the representation overhead required to associate blocks with their shortest or near-shortest explanatory programs. We then show that reuse under the chosen estimators improves this bound by the total reuse gain \(\hat{\Delta}\), and becomes strictly better than BDM\,1.0 whenever \(\hat{\Delta}\) exceeds this baseline representation overhead. Finally, we relate the corresponding theoretical reuse gain \(\Delta\) to algorithmic mutual information, identifying shared algorithmic structure as the source of exploitable reuse.

Recall that \(B(X)=\{x_1,\ldots,x_n\}\) denotes the set of distinct blocks of \(X\), with multiplicities \(m_1,\ldots,m_n\). For each block \(x_i\), let \(\mathcal P(x_i)\) denote the set of candidate programs capable of generating \(x_i\).

We assume that, for each block \(x_i\), the candidate set \(\mathcal P(x_i)\) contains a shortest or near-shortest explanatory program \(p_i^*\) such that \(U(p_i^*)=x_i\). Let \(P^*=\{p_1^*,\ldots,p_n^*\}\). Since \(p_i^*\in\mathcal P(x_i)\) for each \(i\), it follows that \(P^*\in\mathcal E_X\).

For any explanatory set \(P=\{p_1,\ldots,p_n\}\in\mathcal E_X\), define the program--observation discrepancy
\[
c_{\mathrm{rep}}(p_i)
=
\left|
\hat K(x_i)-\hat K(p_i)
\right|,
\]
and its accumulated representation overhead
\[
C_{\mathrm{rep}}(P)
=
\sum_{i=1}^{n}c_{\mathrm{rep}}(p_i).
\]
For the shortest or near-shortest explanatory set \(P^*=\{p_1^*,\ldots,p_n^*\}\), this gives
\[
C_{\mathrm{rep}}(P^*)
=
\sum_{i=1}^{n}
\left|
\hat K(x_i)-\hat K(p_i^*)
\right|.
\]
This quantity measures the discrepancy between the estimated complexities of the blocks and their baseline explanatory programs. The following comparison holds for any finite \(C_{\mathrm{rep}}(P^*)\), but recovering BDM\,1.0 as a close independent-description limit requires this overhead to be small. No analogous small-overhead assumption is imposed on arbitrary \(P\in\mathcal E_X\).

This bounded-overhead condition is consistent with true Kolmogorov complexity. Since \(U(p_i^*)=x_i\), we have \(K(x_i)\leq K(p_i^*)+O(1)\), so \(K(p_i^*)\) cannot be substantially smaller than \(K(x_i)\) without contradicting the minimality of \(K(x_i)\). Conversely, because \(p_i^*\) is shortest or near-shortest, \(K(p_i^*)\leq K(x_i)+O(1)\), up to the bounded excess length of the near-shortest description.

BDM\,1.0 treats each distinct block independently. Since BDM\,2.0 may always choose not to reuse information between explanations, the independent-description strategy of BDM\,1.0 is an admissible special case of BDM\,2.0. The only discrepancy arises from the representation overhead required to relate block descriptions to their shortest or near-shortest explanatory programs.

\begin{theorem}[Independent-explanation upper bound]
\label{thm:bdm2-independent-limit}
Under the assumptions above,
\[
\mathrm{BDM2}(X)
\leq
\mathrm{BDM1}(X)
+
C_{\mathrm{rep}}(P^*).
\]
If the baseline representation discrepancies are uniformly bounded and the number of distinct blocks is bounded by the fixed block vocabulary, then
\[
\mathrm{BDM2}(X)
\leq
\mathrm{BDM1}(X)
+
O(1),
\]
where the \(O(1)\) term is independent of \(|X|\), although its hidden constant may be large.
\end{theorem}

\begin{proof}
BDM\,1.0 describes each distinct block independently:
\[
\mathrm{BDM1}(X)
=
\sum_{i=1}^{n}\hat K(x_i)
+
\sum_{i=1}^{n}\log m_i.
\]
Since \(P^*\in\mathcal E_X\), the choice \(P=P^*\) and \(S=P^*\) is admissible. With this choice, \(P\setminus S=\varnothing\), and therefore the reuse term vanishes. Hence BDM\,2.0 admits the independent description
\begin{equation}
\label{eq:bdm2-independent-bound}
\mathrm{BDM2}(X)
\leq
\sum_{i=1}^{n}\hat K(p_i^*)
+
\sum_{i=1}^{n}\log m_i.
\end{equation}
By definition,
\[
\hat K(p_i^*)
\leq
\hat K(x_i)
+
c_{\mathrm{rep}}(p_i^*).
\]
Therefore,
\[
\begin{aligned}
\mathrm{BDM2}(X)
&\leq
\sum_{i=1}^{n}
\left(
\hat K(x_i)+c_{\mathrm{rep}}(p_i^*)
\right)
+
\sum_{i=1}^{n}\log m_i \\
&=
\sum_{i=1}^{n}\hat K(x_i)
+
\sum_{i=1}^{n}\log m_i
+
\sum_{i=1}^{n}c_{\mathrm{rep}}(p_i^*) \\
&=
\mathrm{BDM1}(X)
+
C_{\mathrm{rep}}(P^*).
\end{aligned}
\]

If \(c_{\mathrm{rep}}(p_i^*)\leq\bar c_\mathrm{rep}\) uniformly and the fixed block vocabulary has cardinality \(V\), then \(n\leq V\) and
\[
C_{\mathrm{rep}}(P^*)
=
\sum_{i=1}^{n}c_{\mathrm{rep}}(p_i^*)
\leq
n\bar c_\mathrm{rep}
\leq
V\bar c_\mathrm{rep}.
\]
Hence \(C_{\mathrm{rep}}(P^*)=O(1)\) with respect to \(|X|\), because \(V\) and \(\bar c_\mathrm{rep}\) are fixed independently of the object size. The hidden constant \(V\bar c_\mathrm{rep}\) may nevertheless be large. Consequently,
\[
\mathrm{BDM2}(X)
\leq
\mathrm{BDM1}(X)
+
O(1).
\]
\end{proof}

Theorem~\ref{thm:bdm2-independent-limit} establishes an independent-description baseline. Whenever explanatory programs or observations share algorithmic information, BDM\,2.0 may replace independent descriptions by conditional descriptions. The resulting reduction in description length is captured by the total estimator-level reuse gain \(\hat{\Delta}\).

\begin{theorem}[Reuse-gain improvement bound]
\label{thm:bdm2-reuse-improvement}
Under the assumptions above, suppose that an admissible BDM\,2.0 explanation uses a set \(R\) of reuse relations. For each \((i,j)\in R\), target \(j\) is described relative to source \(i\), either by recovering \(p_j^*\) from \(p_i^*\) or by recovering \(x_j\) directly from \(x_i\). Each target \(j\) appears at most once in \(R\), and every source index remains outside the target set
\[
T_R
=
\{j:(i,j)\in R\}.
\]
Equivalently, every reuse relation points from a selected representative to a non-selected target.

For any explanatory set \(P=\{p_1,\ldots,p_n\}\in\mathcal E_X\), define the pairwise reuse cost
\[
r_{ij}
=
\min
\left\{
\hat K(p_j\mid p_i),
\hat K(x_j\mid x_i)
\right\},
\]
the corresponding estimator-level reuse gain
\[
\hat{\Delta}_{ij}
=
\hat K(p_j)-r_{ij},
\]
and the total estimator-level reuse gain
\[
\hat{\Delta}
=
\sum_{(i,j)\in R}
\hat{\Delta}_{ij}.
\]
For the shortest-program baseline \(P=P^*\), let
\[
r_{ij}^*
=
\min
\left\{
\hat K(p_j^*\mid p_i^*),
\hat K(x_j\mid x_i)
\right\},
\]
and define
\[
\hat{\Delta}_{ij}^*
=
\hat K(p_j^*)-r_{ij}^*,
\qquad
\hat{\Delta}^*
=
\sum_{(i,j)\in R}
\hat{\Delta}_{ij}^*.
\]
Then
\[
\mathrm{BDM2}(X)
\leq
\mathrm{BDM1}(X)
+
C_{\mathrm{rep}}(P^*)
-
\hat{\Delta}^*.
\]
Consequently, if \(\hat{\Delta}^*>0\), then
\[
\mathrm{BDM2}(X)
<
\mathrm{BDM1}(X)
+
C_{\mathrm{rep}}(P^*).
\]
Moreover, if
\[
\hat{\Delta}^*
>
C_{\mathrm{rep}}(P^*),
\]
then
\[
\mathrm{BDM2}(X)
<
\mathrm{BDM1}(X).
\]
\end{theorem}

\begin{proof}
Let
\[
T_R
=
\{j:(i,j)\in R\}
\]
be the set of reused targets. Consider the choice
\[
P=P^*,
\qquad
S
=
P^*
\setminus
\{p_j^*:j\in T_R\}.
\]
If \(R=\varnothing\), then \(S=P^*\) and the result reduces to Theorem~\ref{thm:bdm2-independent-limit} with \(\hat{\Delta}^*=0\). Otherwise, every reuse relation has a source outside \(T_R\), so \(S\) contains at least one source program. Hence \(S\neq\varnothing\).

For each \((i,j)\in R\), the independent description pays \(\hat K(p_j^*)\), whereas the reused description pays \(r_{ij}^*\). The corresponding estimator-level reuse gain is therefore
\[
\hat{\Delta}_{ij}^*
=
\hat K(p_j^*)-r_{ij}^*.
\]
Summing over all reuse relations gives
\[
\hat{\Delta}^*
=
\sum_{(i,j)\in R}
\hat{\Delta}_{ij}^*.
\]
For the admissible choice above, every source index \(i\) lies outside \(T_R\), so \(p_i^*\in S\). Therefore, for each \((i,j)\in R\), the conditional-description cost of the target \(p_j^*\) satisfies
\[
\min_{p_s^*\in S}
\min
\left\{
\hat K(p_j^*\mid p_s^*),
\hat K(x_j\mid x_s)
\right\}
\leq
\min
\left\{
\hat K(p_j^*\mid p_i^*),
\hat K(x_j\mid x_i)
\right\}
=
r_{ij}^*.
\]
Because \(\mathrm{BDM2}(X)\) minimizes over all admissible explanatory sets and representative subsets, its value is no greater than the cost of this particular reuse construction. Since each target appears at most once in \(R\), summing over the reused targets gives
\[
\mathrm{BDM2}(X)
\leq
\sum_{k\notin T_R}\hat K(p_k^*)
+
\sum_{(i,j)\in R}r_{ij}^*
+
\sum_{k=1}^{n}\log m_k.
\]
By the definition of \(\hat{\Delta}^*\),
\[
\sum_{k\notin T_R}\hat K(p_k^*)
+
\sum_{(i,j)\in R}r_{ij}^*
=
\sum_{k=1}^{n}\hat K(p_k^*)
-
\hat{\Delta}^*.
\]
Hence,
\[
\mathrm{BDM2}(X)
\leq
\sum_{k=1}^{n}\hat K(p_k^*)
+
\sum_{k=1}^{n}\log m_k
-
\hat{\Delta}^*.
\]
This is the independent-description upper bound in Eq.~\eqref{eq:bdm2-independent-bound}, reduced by the total estimator-level reuse gain \(\hat{\Delta}^*\). Applying the same representation-overhead relation used in Theorem~\ref{thm:bdm2-independent-limit},
\[
\hat K(p_k^*)
\leq
\hat K(x_k)
+
c_{\mathrm{rep}}(p_k^*),
\]
we obtain
\[
\begin{aligned}
\mathrm{BDM2}(X)
&\leq
\sum_{k=1}^{n}\hat K(x_k)
+
\sum_{k=1}^{n}\log m_k
+
\sum_{k=1}^{n}c_{\mathrm{rep}}(p_k^*)
-
\hat{\Delta}^* \\
&=
\mathrm{BDM1}(X)
+
C_{\mathrm{rep}}(P^*)
-
\hat{\Delta}^*.
\end{aligned}
\]
Therefore, if \(\hat{\Delta}^*>0\), then
\[
\mathrm{BDM2}(X)
<
\mathrm{BDM1}(X)
+
C_{\mathrm{rep}}(P^*).
\]
Furthermore, if
\[
\hat{\Delta}^*
>
C_{\mathrm{rep}}(P^*),
\]
then
\[
\mathrm{BDM2}(X)
<
\mathrm{BDM1}(X).
\]
\end{proof}

The previous theorem shows that reuse improves the independent-description upper bound whenever \(\hat{\Delta}^*>0\), and that BDM\,2.0 falls below BDM\,1.0 when the shortest-program reuse gain exceeds the baseline representation overhead \(C_{\mathrm{rep}}(P^*)\). In Appendix~\ref{app:bdm2-approximation-error}, we extend this comparison to the true Kolmogorov complexity \(K(X)\), showing that under a common reconstruction overhead \(\Gamma(X)\) sufficient for both estimators,
\[
\hat{\Delta}^*
\geq
C_{\mathrm{rep}}(P^*)
\]
implies
\[
\left|
\bigl(\mathrm{BDM2}(X)+\Gamma(X)\bigr)-K(X)
\right|
\leq
\left|
\bigl(\mathrm{BDM1}(X)+\Gamma(X)\bigr)-K(X)
\right|.
\]

We now relate reuse gain to algorithmic mutual information. The following result defines a theoretical reuse-gain quantity \(\Delta\) directly in terms of the true Kolmogorov complexity \(K\).

\begin{theorem}[Mutual information and total reuse gain]
\label{thm:mutual-information-reuse-gain}
For each reuse relation \((i,j)\in R\), define the theoretical pairwise reuse gains
\[
\Delta_{ij}^{P}
=
K(p_j)-K(p_j\mid p_i),
\qquad
\Delta_{ij}^{X}
=
K(p_j)-K(x_j\mid x_i),
\]
corresponding respectively to program-space reuse and observation-space recovery, both measured relative to the independent program-description cost \(K(p_j)\). We set
\[
\Delta_{ij}
=
\max\{\Delta_{ij}^{P},\Delta_{ij}^{X}\},
\]
and define the total theoretical reuse gain associated with \(R\) by
\[
\Delta
=
\sum_{(i,j)\in R}\Delta_{ij}.
\]
Algorithmic mutual information is defined, up to an additive constant, by
\[
I(a:b)=K(a)-K(a\mid b^*)+O(1),
\]
where \(b^*\) denotes a shortest description of \(b\)~\cite{li_introduction_2019}. For each program \(p_i\), let \(q_i^*\)\footnote{The notation \(q_i^*\) is used to avoid confusion with \(p_i^*\), which denotes the selected shortest or near-shortest generator of \(x_i\) in \(\mathcal P(x_i)\).} denote a shortest description of \(p_i\), and let \(x_i^*\) denote a shortest description of \(x_i\). Let
\[
\eta_i^P=O(\log K(p_i)),
\qquad
\eta_i^X=O(\log K(x_i))
\]
denote the overheads required to compare conditioning on \(p_i\) and \(x_i\) with conditioning on \(q_i^*\) and \(x_i^*\), respectively~\cite{li_introduction_2019}. Then
\[
\Delta
\geq
\sum_{(i,j)\in R}
\max\{
I(p_j:p_i)-\eta_i^P,
I(p_j:x_i)-\eta_i^X
\}
-
O(|R|).
\]
Consequently, for a fixed reuse scheme \(R\), the total theoretical reuse gain is lower-bounded by the accumulated program-program and program-block algorithmic mutual information, up to logarithmic conditioning overhead.
\end{theorem}

\begin{proof}
For each reuse relation \((i,j)\in R\), the standard relation between conditioning on an object and conditioning on its shortest description gives
\[
K(p_j\mid p_i)
\leq
K(p_j\mid q_i^*)+\eta_i^P.
\]
Therefore,
\[
\begin{aligned}
\Delta_{ij}^{P}
&=
K(p_j)-K(p_j\mid p_i) \\
&\geq
K(p_j)-K(p_j\mid q_i^*)-\eta_i^P \\
&=
I(p_j:p_i)-\eta_i^P-O(1).
\end{aligned}
\]

For the observation-space branch, a conditional description of \(p_j\) given \(x_i\) can be followed by executing \(p_j\) to recover \(x_j\). Hence,
\[
K(x_j\mid x_i)
\leq
K(p_j\mid x_i)+O(1).
\]
Moreover,
\[
K(p_j\mid x_i)
\leq
K(p_j\mid x_i^*)+\eta_i^X.
\]
Therefore,
\[
\begin{aligned}
\Delta_{ij}^{X}
&=
K(p_j)-K(x_j\mid x_i) \\
&\geq
K(p_j)-K(p_j\mid x_i)-O(1) \\
&\geq
K(p_j)-K(p_j\mid x_i^*)-\eta_i^X-O(1) \\
&=
I(p_j:x_i)-\eta_i^X-O(1).
\end{aligned}
\]

Since
\[
\Delta_{ij}
=
\max\{\Delta_{ij}^{P},\Delta_{ij}^{X}\},
\]
we obtain
\[
\Delta_{ij}
\geq
\max\{
I(p_j:p_i)-\eta_i^P,
I(p_j:x_i)-\eta_i^X
\}
-
O(1).
\]
Summing over all \((i,j)\in R\) gives
\[
\Delta
\geq
\sum_{(i,j)\in R}
\max\{
I(p_j:p_i)-\eta_i^P,
I(p_j:x_i)-\eta_i^X
\}
-
O(|R|).
\]

\end{proof}

The distinction between \(\hat{\Delta}\) and \(\Delta\) separates the computable estimator-level gain used by BDM\,2.0 from the theoretical gain predicted by algorithmic information. In practice, the improvement bound depends on \(\hat{\Delta}\), while the mutual-information theorem explains when substantial reuse should exist at the theoretical level.

\subsection[Estimator Consistency and Reuse Gains]
{Estimator Consistency and Reuse Gains}
\label{sub:estimator-consistency-reuse-gain}

The previous subsection established three results. First, BDM\,2.0 contains the independent BDM\,1.0 description as an admissible limiting case. Second, estimator-level reuse improves this independent-description bound whenever the total reuse gain \(\hat{\Delta}\) is large enough to exceed the relevant overheads. Third, at the theoretical level, the total reuse gain \(\Delta\) is lower-bounded by accumulated algorithmic mutual information among the reused programs or blocks. We now separate the structural assumptions used in the comparison from the additional estimator-consistency conditions required for the computable quantities \(\hat K\) and \(\hat K(\cdot\mid\cdot)\) to reveal these theoretical reuse gains in practice.

The comparison above relied on two structural assumptions:
\begin{enumerate}
    \item \emph{Candidate-program coverage:} for each block \(x_i\), the candidate set \(\mathcal P(x_i)\) contains a shortest or near-shortest explanatory program \(p_i^*\), with \(p_i^*\in\mathcal P(x_i)\) and \(U(p_i^*)=x_i\). Consequently, \(P^*=\{p_1^*,\ldots,p_n^*\}\) is an admissible element of \(\mathcal E_X\).

    \item \emph{Bounded baseline representation overhead:} for each selected explanatory program \(p_i^*\), the estimated object-level and program-level complexities have the controlled discrepancy
    \[
    c_{\mathrm{rep}}(p_i^*)
    =
    \left|\hat K(x_i)-\hat K(p_i^*)\right|.
    \]
    In CTM-based implementations, this bounded representation overhead can be estimated empirically by an affine alignment between output-level and program-level CTM tables, as described in Appendix~\ref{app:ctm-affine-calibration}.
\end{enumerate}

These assumptions establish a well-defined independent-description baseline for BDM\,2.0. However, they do not by themselves ensure that estimated reuse is meaningful. A conditional estimator should not assign arbitrarily small values to \(\hat K(x_j\mid x_i)\) or \(\hat K(p_j\mid p_i)\) unless the conditioning object contains algorithmic information about the object being described. BDM\,2.0 therefore requires consistency between estimated conditional reductions and the underlying algorithmic mutual information.

We therefore identify two additional estimator-level requirements governing when admissible reuse descriptions yield a reliable improvement over BDM\,1.0:
\begin{enumerate}
    \setcounter{enumi}{2}
    \item \emph{Reuse-validation overhead:} the overhead \(c_{\mathrm{reuse}}\) associated with specifying, selecting, or validating reuse relations must remain small enough that the accumulated raw reuse gain exceeds the combined representation and reuse-validation threshold \(C_{\mathrm{rep}}(P)+|R|c_{\mathrm{reuse}}\). This overhead is not part of the BDM\,2.0 objective; it is a conservative margin used to decide when estimated reuse is large enough to support a reliable improvement claim.

    \item \emph{Mutual-information consistency:} estimated reuse gains should track the corresponding theoretical algorithmic mutual informations up to a fixed scale and bounded error. This condition connects the computable gains used by BDM\,2.0 to the theoretical gains identified in Theorem~\ref{thm:mutual-information-reuse-gain}.
\end{enumerate}

Theorem~\ref{thm:bdm2-reuse-improvement} shows that BDM\,2.0 improves upon BDM\,1.0 whenever the estimator-level reuse gain \(\hat{\Delta}^*\) exceeds the baseline representation overhead \(C_{\mathrm{rep}}(P^*)\). More generally, for any fixed \(P\in\mathcal E_X\) and reuse set \(R\) satisfying the same admissibility conditions as in Theorem~\ref{thm:bdm2-reuse-improvement}, the same construction gives
\[
\mathrm{BDM2}(X)
\leq
\mathrm{BDM1}(X)
+
C_{\mathrm{rep}}(P)
-
\sum_{(i,j)\in R}
\left(
\hat K(p_j)-r_{ij}
\right).
\]
In practice, however, one may require an additional validation margin for specifying, selecting, or accepting reuse relations. The following definition gives a conservative sufficient condition in terms of the raw conditional reuse gain and the accumulated reuse-validation threshold.

\begin{definition}[Reliability threshold for reuse]
\label{def:bdm2-reliability-threshold}
Let \(P=\{p_1,\ldots,p_n\}\in\mathcal E_X\) be an admissible explanatory set, and let \(C_{\mathrm{rep}}(P)\) denote its accumulated program--observation representation overhead. Let \(c_{\mathrm{reuse}}\geq 0\) denote an upper bound on the additional description cost required to realize each reuse relation beyond its estimated conditional cost, including the overhead of identifying the available source, selecting the program- or observation-level reuse path, and invoking the corresponding conditional decoding mechanism. For a fixed coding scheme and bounded candidate vocabulary, these operational costs are accounted for by \(c_{\mathrm{reuse}}\). For a reuse set \(R\) satisfying the same admissibility conditions as in Theorem~\ref{thm:bdm2-reuse-improvement}, define the net reliable reuse gain as
\[
\hat{\Delta}_{\mathrm{rel}}(P,R)
=
\sum_{(i,j)\in R}
\left(
\hat K(p_j)-r_{ij}
\right)
-
|R|c_{\mathrm{reuse}}.
\]
The pair \((P,R)\) satisfies the reliability threshold whenever
\[
\hat{\Delta}_{\mathrm{rel}}(P,R)
>
C_{\mathrm{rep}}(P),
\]
or equivalently,
\[
\sum_{(i,j)\in R}
\left(
\hat K(p_j)-r_{ij}
\right)
>
C_{\mathrm{rep}}(P)
+
|R|c_{\mathrm{reuse}}.
\]
Any pair \((P,R)\) satisfying this threshold guarantees
\[
\mathrm{BDM2}(X)
<
\mathrm{BDM1}(X).
\]
\end{definition}
Definition~\ref{def:bdm2-reliability-threshold} gives a sufficient condition for BDM\,2.0 to improve upon BDM\,1.0. The threshold \(C_{\mathrm{rep}}(P)+|R|c_{\mathrm{reuse}}\) combines the program--observation representation overhead with a conservative allowance for per-relation costs not captured by \(r_{ij}\), including source identification, reuse-path selection, and conditional decoding. Any reuse gain exceeding this threshold guarantees a reduction relative to the independent-description baseline. Because \(c_{\mathrm{reuse}}\) serves only as a certification margin, it is not included in the operational BDM\,2.0 objective.

It remains to specify when the estimated quantities reflect the theoretical information underlying the reuse gains in Theorem~\ref{thm:mutual-information-reuse-gain}. For program-space reuse, define
\[
I_{ij}^{P}
=
I(p_j:p_i).
\]
For observation-space recovery, define the cross-representation quantity
\[
I_{ij}^{PX}
=
I(p_j:x_i).
\]
Direct estimation of \(I_{ij}^{PX}\) is not required to obtain a conservative sufficient signal of shared structure supporting observation-space recovery. Since \(p_j\mapsto U(p_j)=x_j\), the deterministic algorithmic data-processing inequality gives \cite{grunwald_shannon_2004}
\[
I_{ij}^{X}
:=
I(x_j:x_i)
\leq
I_{ij}^{PX}
+
O(1).
\]
Thus, large mutual information between the observed blocks implies comparably large mutual information between the target program and the source block. The converse need not hold, so a small output-level signal does not exclude reuse detectable in program space or only across representations.

Algorithmic mutual information is strictly given by \(I(x:y)=K(x)-K(x\mid y^*)+O(1)\), where \(y^*\) is a shortest description of the conditioning object \(y\). Because \(y^*\) is not computable in general, the operational estimator conditions directly on the available representation \(y\). Replacing \(K(x\mid y^*)\) with \(K(x\mid y)\) introduces at most an \(O(\log K(y))\) discrepancy, which we absorb into the consistency tolerance. We therefore define
\[
\hat I_{ij}^{P}
=
\hat K(p_j)-\hat K(p_j\mid p_i),
\qquad
\hat I_{ij}^{X}
=
\hat K(x_j)-\hat K(x_j\mid x_i).
\]
The program-space signal \(\hat I_{ij}^{P}\) estimates \(I_{ij}^{P}\), whereas the output-space signal \(\hat I_{ij}^{X}\) estimates \(I_{ij}^{X}\) and provides a conservative proxy for the cross-representation quantity \(I_{ij}^{PX}\). 

We formalise this requirement as \emph{mutual-information consistency}. For constants \(\alpha_P,\alpha_X>0\) and bounded errors \(\varepsilon_P,\varepsilon_X\), we require
\[
\left|
\hat I_{ij}^{P}
-
\alpha_P I_{ij}^{P}
\right|
\leq
\varepsilon_P,
\qquad
\left|
\hat I_{ij}^{X}
-
\alpha_X I_{ij}^{X}
\right|
\leq
\varepsilon_X
\]
for the reuse relations under consideration.

This condition prevents arbitrary reductions in estimated conditional complexity from being interpreted as evidence of reuse. When \(I_{ij}^{P}\) or \(I_{ij}^{X}\) is small, the corresponding estimated signal remains small up to bounded error. When it is large, the estimator preserves a proportional signal. Mutual-information consistency therefore ensures that estimated conditional reductions track genuine shared algorithmic information. For observation-space recovery, consistency with \(I_{ij}^{X}\) provides a sufficient but conservative validation of the cross-representation information identified in Theorem~\ref{thm:mutual-information-reuse-gain}.

Together, the reuse-overhead condition and mutual-information consistency give the practical estimator requirement for BDM\,2.0 improvement: estimated conditional reductions must be large enough to overcome representation and conditional-reuse overheads, while remaining supported by corresponding algorithmic mutual information.

\section[Towards Operationalizing BDM 2.0]{Towards Operationalizing BDM\,2.0}

The previous subsections established the conditions under which BDM\,2.0 can exploit reusable explanations. We now outline a concrete pathway for constructing such an estimator using the Coding Theorem Method (CTM).

This construction should be understood as an operational roadmap rather than a completed implementation. The goal is to instantiate the objects appearing in the definition of BDM\,2.0, including \(\mathcal P(x_i)\), \(\mathcal E_X\), \(\hat K(p)\), and \(\hat K(\cdot\mid\cdot)\), from finite CTM tables, and to estimate or bound the implementation-level counterparts of the representation and reuse overheads \(C_{\mathrm{rep}}(P)\) and \(c_{\mathrm{reuse}}\). 

\subsection{Program Representations and Candidate Generators}

Standard CTM estimates the algorithmic complexity of finite objects by enumerating a finite class of programs and recording the frequency with which each output is produced. In the usual BDM\,1.0 setting, the object of interest is the output block \(x\). For BDM\,2.0, however, the generating programs must also be available as objects of analysis. A Turing machine, for example, may be represented by its transition table or by the corresponding state diagram. More generally, any finite program model used by CTM induces a finite representation of programs that can itself be analysed.

Let \(\mathcal M\) denote the finite program class explored by CTM under a fixed representation, size bound, and runtime cutoff. For each block \(x\), define
\[
\mathcal P_{\mathrm{CTM}}(x)
=
\{p\in\mathcal M:U(p)=x\}.
\]
Unlike the abstract candidate set \(\mathcal P(x)\), this set is finite and can be recovered directly from the CTM enumeration. It is nonempty precisely when \(x\) is generated by at least one program in \(\mathcal M\).

Program-level complexity additionally requires a fixed finite encoding \(\operatorname{enc}(p)\) of each \(p\in\mathcal M\). A program-level CTM table can then assign
\[
\mathrm{CTM}_P(p)
=
\mathrm{CTM}\bigl(\operatorname{enc}(p)\bigr),
\]
where the encoding may be based, for example, on the transition table or state diagram of the machine. The same encoding must be used throughout the construction, since finite CTM estimates can depend on the chosen representation.

For each block \(x_i\), select
\[
p_i^*
\in
\arg\min_{p\in\mathcal P_{\mathrm{CTM}}(x_i)}
\mathrm{CTM}_P(p).
\]
Thus, \(p_i^*\) is the simplest generator found within the enumerated class \(\mathcal M\), rather than necessarily a globally shortest program for \(x_i\). Provided that every \(\mathcal P_{\mathrm{CTM}}(x_i)\) is nonempty, the collection
\[
P^*
=
\{p_1^*,\ldots,p_n^*\}
\]
defines a CTM-instantiated explanatory set for the distinct blocks of \(X\).

More generally, the enumeration induces the finite explanatory space
\[
\mathcal E_X^{\mathrm{CTM}}
=
\left\{
\{p_1,\ldots,p_n\}:
p_i\in\mathcal P_{\mathrm{CTM}}(x_i)
\text{ for every }i
\right\}.
\]
Each block is therefore associated with an output-level estimate \(\mathrm{CTM}_X(x_i)\), one or more candidate generators, and their program-level estimates. These quantities provide the finite search space required for program selection and reuse optimization in BDM\,2.0.

\subsection{Simplest-Program--Output Validation}
\label{sec:simplest-program-output-validation}

The previous construction assigns to each block \(x_i\) a simplest available generator \(p_i^*\), whose program-level CTM estimate must be comparable with the CTM estimate of \(x_i\). This compatibility is required for the comparison with the BDM\,1.0 independent-description limit established in Theorem~\ref{thm:bdm2-independent-limit}.

Let \(\mathrm{CTM}_X(x_i)\) denote the output-level estimate and \(\mathrm{CTM}_P(p_i)\) the estimate of a candidate program representation. Because these estimates may be obtained from different finite CTM spaces, their scales need not coincide. We therefore estimate constants \(\alpha>0\) and \(\gamma\), fixed for the chosen representations, and define
\[
\tau_{\mathrm{rep}}(p_i)
=
\left|
\mathrm{CTM}_X(x_i)
-
\bigl(\alpha\,\mathrm{CTM}_P(p_i)+\gamma\bigr)
\right|.
\]
The affine calibration accounts for systematic differences between the output and program CTM scales, while \(\tau_{\mathrm{rep}}(p_i)\) records the remaining representation discrepancy. Its derivation from finite CTM counts is given in Appendix~\ref{app:ctm-affine-calibration}.

For an explanatory set \(P=\{p_1,\ldots,p_n\}\in\mathcal E_X^{\mathrm{CTM}}\), define the accumulated representation tolerance 
\[ 
\mathcal T_{\mathrm{rep}}(P)
=
\sum_{i=1}^{n}\tau_{\mathrm{rep}}(p_i). 
\]
This is the CTM-construction analogue of the representation overhead \(C_{\mathrm{rep}}(P)\). It is not universal, but depends on the program encoding, calibrated CTM tables, and block vocabulary.

For the simplest available explanatory set \(P^*=\{p_1^*,\ldots,p_n^*\}\),
\[
\tau_{\mathrm{rep}}(p_i^*)
=
\left|
\mathrm{CTM}_X(x_i)
-
\bigl(\alpha\,\mathrm{CTM}_P(p_i^*)+\gamma\bigr)
\right|,
\]
and
\[
\mathcal T_{\mathrm{rep}}(P^*)
=
\sum_{i=1}^{n}\tau_{\mathrm{rep}}(p_i^*).
\]
This quantity is the implementation-level counterpart of \(C_{\mathrm{rep}}(P^*)\) in Theorem~\ref{thm:bdm2-independent-limit}. A valid construction should yield uniformly small local discrepancies, and hence a small value of \(\mathcal T_{\mathrm{rep}}(P^*)\). If no such calibration exists, the comparison cannot be supported with a small representation overhead.

In what follows, the calibration is absorbed into the notation. Thus, \(\mathrm{CTM}_P\) and the corresponding conditional CTM tables are understood to be expressed on the calibrated output-space CTM scale, while \(\mathcal T_{\mathrm{rep}}(P)\) records the remaining accumulated program--output discrepancy.

\subsection{Conditional CTM for Reuse Estimation}

BDM\,2.0 requires conditional estimates for describing one block or generator relative to another. A CTM implementation therefore requires the tables
\[
\mathrm{CTM}_{X\mid X}(x_j\mid x_i),
\qquad
\mathrm{CTM}_{P\mid P}(p_j\mid p_i),
\]
which encode reuse in output and program space, respectively. Throughout this section, program-level quantities are understood on the calibrated output-space CTM scale introduced above.

For a fixed explanatory set \(P=\{p_1,\ldots,p_n\}\in\mathcal E_X^{\mathrm{CTM}}\), define the reuse cost
\[
r_{ij}^{\mathrm{CTM}}
=
\min
\left\{
\mathrm{CTM}_{P\mid P}(p_j\mid p_i),
\mathrm{CTM}_{X\mid X}(x_j\mid x_i)
\right\}.
\]
This is the operational counterpart of \(r_{ij}\) in the theoretical formulation and is meaningful only when both conditional tables are expressed on a common CTM scale.

The conditional tables should satisfy a practical analogue of the principle that \emph{information cannot hurt}. Since a conditional program may ignore its auxiliary input, \(K(x\mid y)\leq K(x)+O(1)\), so estimated conditional costs should not exceed their unconditional counterparts beyond implementation tolerance.
\[
\max_{i\neq j}
\mathrm{CTM}_{X\mid X}(x_j\mid x_i)
\approx
\mathrm{CTM}_{X}(x_j),
\qquad
\max_{i\neq j}
\mathrm{CTM}_{P\mid P}(p_j\mid p_i)
\approx
\mathrm{CTM}_{P}(p_j).
\]
After normalization, equality represents an empirically uninformative condition, while smaller values indicate reusable information. This normalization aligns conditional and unconditional estimates within each space, whereas the affine calibration aligns the program and output spaces.

For each reuse relation \((i,j)\), define the program-space reuse gain
\[
\hat{\Delta}_{ij}^{P,\mathrm{CTM}}
=
\mathrm{CTM}_{P}(p_j)
-
\mathrm{CTM}_{P\mid P}(p_j\mid p_i),
\]
and the observation-space recovery gain
\[
\hat{\Delta}_{ij}^{X,\mathrm{CTM}}
=
\mathrm{CTM}_{P}(p_j)
-
\mathrm{CTM}_{X\mid X}(x_j\mid x_i).
\]
Both are measured relative to the same independent program-description baseline. The total CTM reuse gain is
\[
\hat{\Delta}^{\mathrm{CTM}}_R(P)
=
\sum_{(i,j)\in R}
\max
\left\{
\hat{\Delta}_{ij}^{P,\mathrm{CTM}},
\hat{\Delta}_{ij}^{X,\mathrm{CTM}}
\right\}.
\]

Let \(\tau_{\mathrm{reuse}}\geq 0\) bound the additional decoding and composition cost associated with each reuse relation. For a representative set \(S\) and reuse set \(R\), the parameterized CTM construction improves upon the independent program-description baseline whenever
\[
\hat{\Delta}^{\mathrm{CTM}}_R(P)
>
\mathcal T_{\mathrm{rep}}(S)
+
|R|\tau_{\mathrm{reuse}},
\]
where
\[
\mathcal T_{\mathrm{rep}}(S)
=
\sum_{p_i\in S}\tau_{\mathrm{rep}}(p_i).
\]
Here, \(\mathcal T_{\mathrm{rep}}(S)\) accounts for the residual program--output discrepancies of the independently described representatives, while \(\tau_{\mathrm{reuse}}\) accounts for the additional cost of realizing each conditional description.

The calibrated quantities yield the following parameterized CTM implementation of BDM\,2.0:
\[
\begin{aligned}
\mathrm{BDM2}_{\tau_{\mathrm{rep}},\tau_{\mathrm{reuse}}}^{\mathrm{CTM}}(X)
={}&
\min_{P\in\mathcal E_X^{\mathrm{CTM}}}
\min_{\varnothing\neq S\subseteq P}
\Bigg[
\sum_{p_i\in S}
\left[
\mathrm{CTM}_{P}(p_i)
+
\tau_{\mathrm{rep}}(p_i)
\right]
\\
&\quad+
\sum_{p_j\in P\setminus S}
\Bigg[
\min_{p_i\in S}
\min
\left\{
\begin{aligned}
&\mathrm{CTM}_{P\mid P}(p_j\mid p_i),\\
&\mathrm{CTM}_{X\mid X}(x_j\mid x_i)
\end{aligned}
\right\}
+
\tau_{\mathrm{reuse}}
\Bigg]
\Bigg]
\\
&\quad+
\sum_{i=1}^{n}\log m_i.
\end{aligned}
\]
Here, \(\tau_{\mathrm{rep}}(p_i)\) penalizes the residual program--output discrepancy of each independently described representative, while \(\tau_{\mathrm{reuse}}\) penalizes each conditional description. Setting both parameters to zero recovers the direct CTM instantiation of BDM\,2.0.

\subsection{Mutual-Information Consistency Checks}

We assume that CTM provides valid finite approximations to Kolmogorov complexity over the chosen output and program domains. Under this premise, the following checks determine whether conditional and joint CTM constructions recover compatible estimates of algorithmic mutual information. The joint construction provides an independent reference derived from marginal and joint CTM values rather than a ground-truth quantity. All quantities below are understood on the calibrated CTM scale introduced above.

Theorem~\ref{thm:mutual-information-reuse-gain} relates observation-space recovery to the cross-representation quantity \(I(p_j:x_i)\). As established in Section~\ref{sub:estimator-consistency-reuse-gain}, \(I(x_j:x_i)\) provides a conservative sufficient signal for this shared structure. We use \(I(x_j:x_i)\) because it is defined entirely within the output domain and directly corresponds to the conditional term \(\mathrm{CTM}_{X\mid X}(x_j\mid x_i)\) used by the BDM\,2.0 objective. It can therefore be compared with an output-level joint CTM estimate without requiring a joint estimator over program--output pairs.

Let \(Z\in\{X,P\}\) denote either the output or program representation, with \(z_i=x_i\) when \(Z=X\) and \(z_i=p_i\) when \(Z=P\). Define the conditional mutual-information estimate
\[
\hat I_{ij}^{Z,\mathrm{cond}}
=
\mathrm{CTM}_{Z}(z_j)
-
\mathrm{CTM}_{Z\mid Z}(z_j\mid z_i),
\]
and, when a joint CTM value is available, define
\[
\hat I_{ij}^{Z,\mathrm{joint}}
=
\mathrm{CTM}_{Z}(z_i)
+
\mathrm{CTM}_{Z}(z_j)
-
\mathrm{CTM}_{Z,Z}(z_i,z_j).
\]

The joint complexity may be evaluated by encoding the ordered pair as a single object:
\[
\mathrm{CTM}_{Z,Z}(z_i,z_j)
=
\mathrm{CTM}_{Z}\bigl(\langle z_i,z_j\rangle\bigr),
\]
where \(\langle z_i,z_j\rangle\) is a fixed uniquely decodable pairing of the corresponding representations. When the same enumeration and normalization are used, the marginal and joint quantities lie on the same empirical CTM scale. For equal-sized output blocks, the pairing may be implemented by direct concatenation because the boundary is known.

Conditional estimates are directional, whereas joint mutual information is symmetric. Consistency may therefore be assessed separately for \(i\to j\) and \(j\to i\), or after symmetrising the conditional estimates. For each \(Z\in\{X,P\}\), we require
\[
\hat I_{ij}^{Z,\mathrm{cond}}
\approx
\hat I_{ij}^{Z,\mathrm{joint}}
\]
over the reuse relations under consideration, up to finite-enumeration, calibration, and conditioning tolerances.

These checks determine whether the conditional and joint constructions preserve the same algorithmic mutual-information signal. Agreement supports the interpretation of conditional reductions as reusable algorithmic information, while systematic disagreement indicates inconsistency among the corresponding marginal, conditional, or joint CTM estimates.

\subsection{Computational Cost of Reuse Optimization}
\label{sub:computational-cost}

Once the tables for the complexity estimator \(\hat K\), such as calibrated CTM tables, have been computed, BDM\,1.0 is linear in the number of distinct blocks. If \(X\) contains \(n\) distinct blocks, BDM\,1.0 requires one lookup per distinct block and therefore has cost \(O(n)\) after preprocessing.

BDM\,2.0 replaces this independent lookup by an optimization over the admissible explanation family \(\mathcal E_X\). For each explanatory set \(P\in\mathcal E_X\), the method selects a representative subset
\[
S\subseteq P
\]
whose programs are paid for directly, while programs in \(P\setminus S\) are described conditionally from representatives in \(S\). Given precomputed reuse costs \(r_{ij}\), a fixed \(S\) can be evaluated by assigning each non-selected program to its cheapest selected source. Exact search over \(S\), however, may require considering \(2^{|P|}-1\) nonempty subsets. In the full problem, this cost is compounded by the optimization over \(\mathcal E_X\).

For general precomputed unconditional and reuse costs, exact representative selection is NP-hard. It suffices to consider the restricted case in which \(\mathcal E_X\) contains a single explanatory set \(P\). The reduction is from weighted dominating set~\cite{garey_computers_1979}. Given a graph \(G=(V,E)\) with non-negative vertex weights \(w_i\), the goal is to find a minimum-weight subset \(S\subseteq V\) such that every vertex either belongs to \(S\) or is adjacent to a vertex in \(S\).

Associate each vertex \(i\in V\) with a program \(p_i\in P\), and set
\[
\hat K(p_i)=w_i.
\]
For distinct programs, define
\[
r_{ij}
=
\begin{cases}
0, & \text{if } \{i,j\}\in E,\\
M, & \text{otherwise},
\end{cases}
\qquad
M>\sum_i w_i.
\]
Because a non-adjacent assignment costs more than selecting every program directly, no optimal solution uses such an assignment. Every program must therefore be selected or assigned to a selected adjacent program, so \(S\) forms a dominating set of \(G\). Since all permitted reuse assignments have zero cost, the BDM\,2.0 objective reduces to
\[
w(S)=\sum_{i\in S}w_i,
\]
which is exactly the weighted dominating-set objective. Exact reuse optimization is therefore NP-hard. This reduction establishes hardness for general precomputed cost matrices, but not necessarily for the narrower class of matrices induced by a particular CTM representation.

For a fixed CTM representation, the optimization is finite but remains combinatorial. For instance, in the two-dimensional CTM setting~\cite{zenil_two-dimensional_2015}, binary \(4\times4\) blocks admit at most \(2^{16}\) configurations. If the search is restricted to the canonical shortest or near-shortest generators \(P=P^*\), which are admissible by Theorem~\ref{thm:bdm2-independent-limit}, then \(|P^*|\leq2^{16}\), and \(|P^*|=n\) for an object with \(n\) distinct blocks. This reduces the candidate vocabulary relative to the full family \(\mathcal E_X\), but does not remove the exponential dependence on the number of active explanatory programs.

In practice, one may restrict \(\mathcal E_X\) to smaller admissible families, such as \(P=P^*\), or use computational strategies including pruning, integer linear programming, branch-and-bound, local search, and graph-based approximations.

One practical heuristic is to select candidate reuse edges from the original weighted costs before optimizing the representative set. For a fixed explanatory set \(P\), define the thresholded reuse graph by
\[
\mathcal R_{\mathrm{thr}}
=
\left\{
(i,j):
\hat K(p_j)-r_{ij}\geq\tau_{\mathrm{gain}}
\right\},
\]
where \(\tau_{\mathrm{gain}}\geq0\) is the minimum estimated gain required to retain a reuse relation. Thus, an edge from \(p_i\) to \(p_j\) is retained only when describing \(p_j\) from \(p_i\) saves at least \(\tau_{\mathrm{gain}}\) relative to describing \(p_j\) independently. Constructing \(\mathcal R_{\mathrm{thr}}\) requires \(O(|P|^2)\) pairwise checks.

Representative selection on the retained graph uses the same binary-edge formulation as in the NP-hardness reduction above, assigning zero cost to retained relations and excluding all others. This yields a weighted set-cover problem with representative weight
\[
w_i=\hat K(p_i).
\]

The standard greedy algorithm selects the representative with the lowest weight per newly covered program. A direct implementation runs in
\[
O\bigl(|P|^2+|\mathcal R_{\mathrm{thr}}|\bigr)
=
O(|P|^2),
\]
and satisfies~\cite{chvatal_greedy_1979}
\[
w(S_{\mathrm{greedy}})
\leq
\bigl(1+\ln(d_{\mathrm{thr}}+1)\bigr)
w(S_{\mathrm{opt}}),
\]
where
\[
d_{\mathrm{thr}}
=
\max_i
\left|
\left\{
j:(i,j)\in\mathcal R_{\mathrm{thr}}
\right\}
\right|
\]
is the maximum out-degree of the thresholded reuse graph. Hence, \(d_{\mathrm{thr}}+1\) is the largest number of programs covered by any representative, including itself. Here, \(S_{\mathrm{opt}}\) is optimal for the thresholded problem. This guarantee does not extend to the full BDM\,2.0 objective with general nonzero reuse costs \(r_{ij}\).

\section{Limitations}
\label{sec:limitations}

As discussed in Section~\ref{sec:bdm1-independent-limit}, the theoretical comparison between BDM\,2.0 and BDM\,1.0 is overhead-sensitive. Reuse improves the independent-description limit only when the savings from conditional descriptions exceed the cost of specifying the selected explanations and reuse relations. These costs are constant once the representation is fixed, but they can be non-negligible at the finite scales at which CTM-based estimators are applied.

A related limitation is estimator validity. BDM\,2.0 depends on unconditional and conditional estimates that remain calibrated to the same algorithmic scale. If \(\hat K(\cdot\mid\cdot)\) is too large, genuine dependencies may be missed and the method reduces toward BDM\,1.0. If it is too small, spurious reuse may be introduced and complexity may be underestimated. Thus, the method requires conditional estimators whose values are neither arbitrary nor merely similarity-based, but aligned with the intended description-length interpretation.

BDM\,2.0 also inherits the dependence of BDM\,1.0 on the chosen decomposition. For a fixed non-overlapping partition, objects with the same block multiset can receive the same value even when different block arrangements have different global algorithmic structure. This is an equivalence induced by the estimator, not an isomorphism of the objects themselves. Overlapping decompositions and boundary corrections reduce this dependence. In the original BDM analysis, the corresponding errors are shown to be bounded, convergent, and asymptotically negligible under the tested boundary strategies \cite{zenil_decomposition_2018}.

For large objects evaluated with a fixed block size and a fixed CTM table, BDM\,2.0 retains the worst-case entropy limit of BDM\,1.0. When no reusable algorithmic information is detected, BDM\,2.0 reduces to the independent-description regime. As the local CTM contribution remains bounded and multiplicity terms dominate, the estimate approaches a block-entropy-like description, as in the original BDM convergence result \cite{zenil_decomposition_2018}.

The current BDM\,2.0 objective is restricted to a single-level explanatory hierarchy. Each non-selected explanation is described directly from a representative program in \(S\), yielding depth-one reuse relations of the form
\[
p_i \longrightarrow p_j,
\qquad p_i\in S.
\]
It does not permit chains
\[
p_1 \xrightarrow{c_{12}} p_2
\xrightarrow{c_{23}} \cdots
\xrightarrow{c_{n-1,n}} p_n,
\]
where a reconstructed program becomes the source of further reuse. Such hierarchies would more closely reflect the chain rule for joint Kolmogorov complexity~\cite{li_introduction_2019},
\[
K(p_1,\ldots,p_n)
=
K(p_1)
+
\sum_{j=2}^{n}
K\!\left(p_j \mid p_1^*,\ldots,p_{j-1}^*\right)
+
O(1),
\]
for fixed \(n\), where \((p_1,\ldots,p_{j-1})^*\) denotes a shortest description of the preceding tuple. In the current formulation, \(p_j\) can serve as a source only if \(p_j\in S\), in which case its unconditional description cost is paid. 

The same restriction applies between programs and outputs. The formulation considers only direct generators \(U(p_i)=x_i\), rather than multi-level descriptions such as
\[
U(h_i^{(\ell)})=h_i^{(\ell-1)}
\quad \text{for } \ell=2,\ldots,n,
\qquad
U(h_i^{(1)})=p_i,
\qquad
U(p_i)=x_i.
\]
For true Kolmogorov complexity, intermediate programs are already implicit in the minimization over all generators of \(x_i\). For a computable estimator, however, an explicit multi-level search may reveal generative structure and reuse that remain inaccessible to a single-level formulation. Extending BDM\,2.0 to chained reuse and multi-level generative descriptions is left for future work.

Finally, as discussed in Section~\ref{sub:computational-cost}, BDM\,2.0 adds an optimization problem over selected explanations and reuse relations. Exact optimization is NP-hard, so practical implementations must rely on restricted candidate sets, pruning, or approximation.

\section{Conclusion}

We introduced BDM\,2.0, a reuse-based extension of the Block Decomposition Method (BDM) in which distinct blocks are no longer forced to pay independent description costs. When no reliable reuse is available, BDM\,2.0 reduces to the BDM\,1.0 regime, up to representation overhead. When blocks or their generators share algorithmic information, BDM\,2.0 can replace independent descriptions with conditional ones and reduce the total code length. The improvement is therefore tied to estimated algorithmic mutual information, not to statistical similarity alone. Algorithmic attention is the induced selection principle, assigning descriptive resources to the components from which the largest compression gain can be obtained.

The main contribution is theoretical. We established that BDM\,1.0 is an admissible limit of BDM\,2.0, derived the overhead condition under which reuse improves the estimate, and related reuse gains to algorithmic mutual information at the level of observations and programs. Realizing this framework in practice requires CTM-derived unconditional and conditional tables for outputs and programs, calibrated on a common scale and tested against systems with known shared generative structure. Such validation will determine when conditional reductions are meaningful, when the estimator should reduce to BDM\,1.0, and when reuse yields a closer approximation to Kolmogorov complexity.

BDM made algorithmic complexity practical beyond the scale of CTM by combining local algorithmic estimates with statistical aggregation. BDM\,2.0 removes this remaining asymmetry. It keeps the local CTM-based estimates, but treats dependencies among components algorithmically as well. This brings decomposition-based complexity estimation closer to the object of AIT itself, where structure is not only what can be described shortly, but also what can be reused across descriptions.

\bibliographystyle{plain}
\bibliography{references}

\newpage

\appendix

\section*{Supplementary Information}

\section{CTM State-Space Counts}
\label{app:ctm_counts}

The Coding Theorem Method (CTM) relies on exhaustive enumeration of small Turing machines. For machines with \(n\) states and \(k\) symbols, each state-symbol pair requires either:
\begin{itemize}
\item a non-halting transition consisting of one new state (chosen from \(n\) options), one new symbol (chosen from \(k\) options), and one head move (chosen from 2 options), or
\item a halting transition consisting of one new symbol (chosen from \(k\) options).
\end{itemize}

Thus, the number of possible instructions is \(2kn+k=k(2n+1)\). Since there are \(n\cdot k\) state-symbol pairs, the total number of machines is
\[
N(n,k)
=
\bigl(k(2n+1)\bigr)^{nk}.
\]
For the common case of \(k=2\),
\[
N(n,2)
=
(4n+2)^{2n}.
\]

For \(n=5\), this yields approximately \(2.7\times 10^{13}\) machines, which can be exhaustively simulated. For \(n=6\), the number grows to approximately \(9.5\times 10^{16}\), requiring partial enumeration. Thus, CTM is feasible only for small \(n\), motivating the need for BDM.

\section[Approximation-Error Bound for BDM 2.0]{Approximation-Error Bound for BDM\,2.0}
\label{app:bdm2-approximation-error}

We show that the reuse-improvement bound implies an approximation-improvement statement with respect to \(K(X)\), provided the reconstruction costs of both estimators are included.

For a fixed decomposition, the original BDM analysis gives
\[
K(X)
\leq
\mathrm{BDM1}(X)+\Gamma_1(X),
\]
where
\[
\Gamma_1(X)
=
O\bigl(\log_2|\mathcal A|\bigr)+e_1(X),
\]
\(\mathcal A\) is the set of admissible arrangements of the decomposed components, and \(e_1(X)\) is the accumulated CTM approximation error~\cite{zenil_decomposition_2018}.

\begin{lemma}[BDM\,2.0 upper-description property]
\label{lem:bdm2-upper-description}
Assume that the unconditional and conditional estimates selected by BDM\,2.0 correspond to realizable prefix-free descriptions whose total length exceeds the corresponding estimated costs by at most \(e_2(X)\). Then
\[
K(X)
\leq
\mathrm{BDM2}(X)+\Gamma_2(X),
\]
where
\[
\Gamma_2(X)
=
O\bigl(\log_2|\mathcal A|\bigr)
+
\omega_2(X)
+
e_2(X)
+
O(1),
\]
and \(\omega_2(X)\) is the cost of specifying the representative set, the reuse assignments, the selected program- or observation-space descriptions, and the required delimiters.
\end{lemma}

\begin{proof}
Consider an admissible pair \(P,S\) attaining the BDM\,2.0 minimum. The programs in \(S\) are recovered from their unconditional descriptions and executed to obtain their corresponding blocks. For each \(p_j\in P\setminus S\), the selected conditional description either recovers \(p_j\) from some \(p_i\in S\), after which \(p_j\) is executed to obtain \(x_j\), or recovers \(x_j\) directly from \(x_i\).

Because every conditional description uses a source in \(S\), all source programs and their corresponding blocks can be recovered before any target is decoded. The representative set, the source assigned to each target, and the choice between program- and observation-space descriptions are specified within \(\omega_2(X)\). The multiplicities and arrangement information then determine how the recovered blocks are assembled to reconstruct \(X\).

The BDM\,2.0 objective accounts for the unconditional and conditional estimated costs and the multiplicity terms. By assumption, replacing the selected estimates by realizable prefix-free descriptions increases their total length by at most \(e_2(X)\). The remaining structural, arrangement, delimiting, and decoding costs are included in \(\Gamma_2(X)\). Therefore,
\[
\mathrm{BDM2}(X)+\Gamma_2(X)
\]
upper-bounds the length of a valid prefix-free description of \(X\). Since \(K(X)\) is the length of the shortest valid prefix-free description of \(X\),
\[
K(X)
\leq
\mathrm{BDM2}(X)+\Gamma_2(X).
\]
\end{proof}

Define the common reconstruction overhead
\[
\Gamma(X)
=
\max\{\Gamma_1(X),\Gamma_2(X)\}.
\]
Although \(\Gamma_2(X)\) includes the additional term \(\omega_2(X)\), it need not dominate \(\Gamma_1(X)\), since \(e_1(X)\), \(e_2(X)\), and the constants hidden in the asymptotic terms may differ. Taking the maximum therefore guarantees that the same reconstruction overhead is sufficient for both estimators. It follows that
\[
K(X)
\leq
\mathrm{BDM1}(X)+\Gamma(X)
\]
and
\[
K(X)
\leq
\mathrm{BDM2}(X)+\Gamma(X).
\]

\begin{proposition}[Approximation improvement under common reconstruction overhead]
\label{prop:bdm2-closer-to-k}
Assume the hypotheses of Theorem~\ref{thm:bdm2-reuse-improvement} and Lemma~\ref{lem:bdm2-upper-description}, and let \(\Gamma(X)\) be the common reconstruction overhead defined above. If
\[
\hat{\Delta}^{*}
\geq
C_{\mathrm{rep}}(P^*),
\]
then
\[
\left|
\bigl(\mathrm{BDM2}(X)+\Gamma(X)\bigr)-K(X)
\right|
\leq
\left|
\bigl(\mathrm{BDM1}(X)+\Gamma(X)\bigr)-K(X)
\right|.
\]
\end{proposition}

\begin{proof}
By Theorem~\ref{thm:bdm2-reuse-improvement},
\[
\mathrm{BDM2}(X)
\leq
\mathrm{BDM1}(X)
+
C_{\mathrm{rep}}(P^*)
-
\hat{\Delta}^{*}.
\]
If
\[
\hat{\Delta}^{*}
\geq
C_{\mathrm{rep}}(P^*),
\]
then
\[
\mathrm{BDM2}(X)
\leq
\mathrm{BDM1}(X).
\]
Adding the common reconstruction overhead gives
\[
\mathrm{BDM2}(X)+\Gamma(X)
\leq
\mathrm{BDM1}(X)+\Gamma(X).
\]
By Lemma~\ref{lem:bdm2-upper-description}, the BDM\,1.0 upper-description property, and the definition of \(\Gamma(X)\), both quantities are valid upper descriptions of \(X\). Hence,
\[
\begin{aligned}
\left|
\bigl(\mathrm{BDM2}(X)+\Gamma(X)\bigr)-K(X)
\right|
&=
\mathrm{BDM2}(X)+\Gamma(X)-K(X)
\\
&\leq
\mathrm{BDM1}(X)+\Gamma(X)-K(X)
\\
&=
\left|
\bigl(\mathrm{BDM1}(X)+\Gamma(X)\bigr)-K(X)
\right|.
\end{aligned}
\]
\end{proof}

The proposition is a common-overhead comparison. The term \(\Gamma(X)\) assigns no additional advantage to BDM\,2.0, since the same overhead is added to both estimators. Once both quantities are valid upper descriptions of \(X\), any estimator-level reuse gain satisfying
\[
\hat{\Delta}^{*}
\geq
C_{\mathrm{rep}}(P^*)
\]
makes the BDM\,2.0 description no farther from \(K(X)\) than the independent BDM\,1.0 description.

\section{Affine Calibration of Finite CTM Scales}
\label{app:ctm-affine-calibration}

This appendix explains the affine calibration used in Section~\ref{sec:simplest-program-output-validation}. The calibration is used to place output-level and program-level CTM estimates on a common empirical scale. It does not assume that two finite CTM tables are theoretically equivalent.

A finite CTM estimate is obtained from empirical output frequencies. If an enumeration \(A\) contains \(N_A\) machines and \(n_x^A\) of them produce \(x\), then
\[
\mathrm{CTM}_A(x)
=
-\log_2\frac{n_x^A}{N_A}
=
\log_2N_A-\log_2n_x^A.
\]
For another enumeration \(B\), with \(N_B\) machines and \(n_x^B\) machines producing \(x\), we have
\[
\mathrm{CTM}_B(x)
=
-\log_2\frac{n_x^B}{N_B}
=
\log_2N_B-\log_2n_x^B.
\]
Therefore,
\[
\mathrm{CTM}_B(x)-\mathrm{CTM}_A(x)
=
\log_2\frac{N_B}{N_A}
-
\log_2\frac{n_x^B}{n_x^A}.
\]
This expression separates two effects. The first term depends only on the change in total enumeration size. The second term depends on how many additional descriptions of the specific object \(x\) are found.

If the larger enumeration changes all object counts by approximately the same factor,
\[
n_x^B
\approx
c\,n_x^A
\]
for all relevant \(x\), then
\[
\mathrm{CTM}_B(x)-\mathrm{CTM}_A(x)
\approx
\log_2\frac{N_B}{N_A}
-
\log_2c.
\]
The difference is then approximately independent of \(x\), so the two CTM tables differ mainly by an additive offset. In this ideal case, an offset parameter \(\gamma\) would be sufficient.

Finite CTM tables usually do not behave so uniformly. Increasing the enumeration budget or changing the model class may add many descriptions for some objects and few descriptions for others. The ratio
\[
\frac{n_x^B}{n_x^A}
\]
then depends on \(x\), so the discrepancy between the two CTM estimates becomes object-dependent. A single additive correction cannot remove this mismatch.

For this reason, we use the affine relation
\[
\mathrm{CTM}_X(x_i)
\approx
\alpha\,\mathrm{CTM}_P(p_i^*)+\gamma
\]
as an empirical calibration between the output-level CTM table and the program-level CTM table. The offset \(\gamma\) accounts for additive shifts caused by finite enumeration effects. The slope \(\alpha\) allows for systematic scale differences between the two CTM spaces, for example differences caused by program encoding, output representation, or model class.

The residual discrepancy is
\[
\tau_{\mathrm{rep}}(p_i^*)
=
\left|
\mathrm{CTM}_X(x_i)
-
\bigl(\alpha\,\mathrm{CTM}_P(p_i^*)+\gamma\bigr)
\right|.
\]
The residual \(\tau_{\mathrm{rep}}(p_i^*)\) measures the remaining program--output mismatch after calibration. Small residuals indicate that the chosen program representation is compatible with the output CTM scale. Large or unstable residuals indicate that the program-level CTM table does not provide a small representation overhead relative to the output-level CTM table.

For the simplest available explanatory set \(P^*=\{p_1^*,\ldots,p_n^*\}\), the accumulated representation tolerance is
\[
\mathcal T_{\mathrm{rep}}(P^*)
=
\sum_{i=1}^{n}
\tau_{\mathrm{rep}}(p_i^*).
\]

The calibration is therefore an implementation-level alignment step. It defines a common empirical CTM scale for subsequent BDM\,2.0 computations, while the remaining mismatch is explicitly recorded by \(\mathcal T_{\mathrm{rep}}(P^*)\).

\end{document}